\documentclass[a4paper,twocolumn,11pt,unpublished]{quantumarticle}
\pdfoutput=1
\usepackage[utf8]{inputenc}
\usepackage[english]{babel}
\usepackage[T1]{fontenc}
\usepackage{amsmath}
\usepackage{hyperref}
\usepackage[numbers,sort&compress]{natbib}
\usepackage{tikz}
\usepackage{lipsum}
\usepackage{bm}
\usepackage{amssymb}
\usepackage{subfigure}
\usepackage{amsthm}

\newtheorem{theorem}{Theorem}

\def\ket#1{| #1 \rangle}
\def\bra#1{\langle #1 |}

\begin{document}

\title{A universal duplication-free quantum neural network}

\author{Xiaokai Hou}
    \affiliation{Institute of  Fundamental and Frontier Sciences, University of Electronic Science and Technology of China, Chengdu, Sichuan, 610051, China}

\author{Guanyu Zhou}
    \email{zhoug@uestc.edu.cn}
    \affiliation{Institute of  Fundamental and Frontier Sciences, University of Electronic Science and Technology of China, Chengdu, Sichuan, 610051, China}

\author{Qingyu Li}
\affiliation{Institute of  Fundamental and Frontier Sciences, University of Electronic Science and Technology of China, Chengdu, Sichuan, 610051, China}

\author{Shan Jin}
\affiliation{Institute of  Fundamental and Frontier Sciences, University of Electronic Science and Technology of China, Chengdu, Sichuan, 610051, China}

\author{Xiaoting Wang}
\email{xiaoting@uestc.edu.cn}
\affiliation{Institute of  Fundamental and Frontier Sciences, University of Electronic Science and Technology of China, Chengdu, Sichuan, 610051, China}
\maketitle

\begin{abstract}
Universality of neural networks describes the ability to approximate arbitrary function, and is a key ingredient to keep the method effective. The established models for universal quantum neural networks(QNN), however, require the preparation of multiple copies of the same quantum state to generate the nonlinearity, with the copy number increasing significantly for highly oscillating functions, resulting in a huge demand for a large-scale quantum processor. To address this problem, we propose a new QNN model that harbors universality without the need of multiple state-duplications, and is more likely to get implemented on near-term devices. To demonstrate the effectiveness, we compare our proposal with two popular QNN models in solving typical supervised learning problems. We find that our model requires significantly fewer qubits and it outperforms the other two in terms of accuracy and relative error. 
\end{abstract}

\section{Introduction}
As an important subfield in machine learning(ML), neural networks(NNs), especially deep NNs, have generated a series of impactful results in many application scenarios~\cite{kalchbrenner2014convolutional,mikolov2011extensions,bojarski2016end,bishop2006pattern}. One of the most striking features of NNs is their ability to learn the hidden patterns of a given data set and to make reliable predictions based on these patterns~\cite{goodfellow2016deep}. Such feature originates from the capability to approximate any continuous function, and is known as the universality. Most NNs proposed in literature are proved to be universal and such results are called universal approximation theorems~\cite{citeulike:3561150,hornik1991approximation}. Due to the power of quantum computation, the idea of quantum machine learning(QML) is proposed to implement ML on quantum circuits, in order to achieve computational advantage compared to the classical counterparts. Such advantage has been shown for many QML algorithms~\cite{biamonte2017quantum}, including quantum support vector machine~\cite{rebentrost2014quantum}, $k$-means clustering~\cite{lloyd2013quantum}, quantum principle component analysis~\cite{lloyd2014quantum}, quantum data-fitting algorithm~\cite{wiebe2012quantum} and quantum Boltzmann machine~\cite{PhysRevX.8.021050}, but not for QNNs. It is shown that certain QNNs have a distinctive prediction advantage on certain designed data sets~\cite{Huang_2021}, but less is known for the general case. In fact, research on the advantage of QNNs is still in progress, partly due to the reason that even the complexity of classical NN algorithms has not been addressed without controversy.

Besides the quantum advantage issue, universality is also crucial to keep QNNs effective. The universality of classical NNs is determined by the nonlinearity of neurons. When it comes to the QNNs, how to generate nonlinearity is one of the biggest impediments to achieve the universal QNN. To address the problem, different QNN models have been proposed such as the continuous variable quantum neural network~\cite{killoran2019continuous}, the quantum neuron~\cite{ventura1998artificial,cao2017quantum}, the circuit-centric quantum classifiers algorithm~\cite{schuld2020circuit} and the quantum circuit learning algorithm~\cite{mitarai2018quantum}, but not all of them have been rigorously proved to be universal. In some of these proposals, nonlinearity relies on using multiple copies of the quantum data, resulting in a rapid increase of the size of the quantum register. In order to solve this problem, in this work, we propose a duplication-free QNN structure which also guarantees the universality of the neural network.

 In this work, we aim to construct a universal QNN model without the need of multiple duplications of quantum data. We design the duplication-free quantum neural network (DQNN) whose nonlinearity is generated by the classical sigmoid function. We further compare the DQNN with two well-known QNN models, the circuit-centric quantum classifiers (CCQ) algorithm~\cite{schuld2020circuit} and the quantum circuit learning (QCL) algorithm~\cite{mitarai2018quantum} in terms of the circuit complexity and the performance on the supervised learning tasks. The results show that the DQNN with fewer qubits outperforms the other two in terms of accuracy and relative error. Besides that, the DQNN has the ability to find the complexity pattern hidden in the real-world data sets and the quantum phase recognition (QPR) task.


\section{DQNN and its structure} 
The universality of DQNN refers to the ability to learn a target function, $f$, hidden in a given data set $D=\{(\boldsymbol{x}_i,y_i)\}^m_{i=1}$ where $\bm{x}_i\in G \subset \mathbb{R}^d$ and $y_i$ is determined by the target function with the data noise, $y_i=f(\bm{x}_i)+\epsilon_i$. The goal of the DQNN is to appropriate the function $f$ using $D$. To achieve this goal, the DQNN uses the structure as shown in Fig.~\ref{frame}. It consists of three parts, a quantum processor (QP), a classical processor (CP) and a classical optimizer (CO). QP part is a parameterized quantum circuit and its output is the expectation values of some measurement observables. CP part contains some parameterized sigmoid function and a linear transformation. CO minimizes a loss function by using the gradient of the parameters in QP and CP.

Before implementing the loop of the three parts, we need to encode the classical data into quantum system using the amplitude encoding method~\cite{PhysRevA.83.032302}. We firstly find a continuous injection, $F$, mapping $\bm{x}_i$ to the $2^n$-dim quantum state Hilbert space $\mathbb{C}_2^{\otimes n}$ with $d<2^n$. If $0\notin G$, we can transform the input $\bm{x}$ into $|\bm{\bar{x}}\rangle\in \mathbb{S}_0\subset \mathbb{C}_2^{\otimes n}$ as
\begin{equation}\label{mapping}
    F:\bm{x}\in G\rightarrow|\bm{\bar{x}}\rangle \equiv\frac{1}{\gamma}(x_1,...,x_d,\tilde{x},0,...,0)^{T}
\end{equation}
where $\tilde{x}=\frac{|\boldsymbol{x}|}{1+|\boldsymbol{x}|}$ and $\gamma=(|\boldsymbol{x}|^2+\tilde{x}^2)^{\frac{1}{2}}$; if $0\in G$, we can perform a shift transformation, $\bm{x}\rightarrow \bm{x}+\alpha$, such that $0\notin G$ which is a ring domain $\{\boldsymbol{x}\in \mathbb{R}^d| 0<\kappa_1\leq |\boldsymbol{x}|\leq \kappa_2\}$. It is worthwhile to mention that $0< \kappa_1 \leq |\boldsymbol{x}| \leq \kappa_2$ implies $\frac{\kappa_1}{1+\kappa_1}\leq \tilde{x} \leq \frac{\kappa_2}{1+\kappa_2}$ and 
\begin{equation}\label{Eq:x-d+1}
(1+(1+\kappa_2)^2)^{-\frac{1}{2}}<\bar{x}_{d+1}<(1+(1+\kappa_1)^2)^{-\frac{1}{2}}
\end{equation}
with $\bar{x}_{d+1}=\frac{\tilde{x}}{\gamma}$. After obtaining the new data set $\{|\bm{\bar{x}}_i\rangle,y_i\}$, we implement the loop of QP, CP and CO to approximate $f$.

\begin{figure}
    \setlength{\belowcaptionskip}{-5.0mm}
    \centering
    \subfigure[]{\includegraphics[height=2.7cm,width=4.7cm]{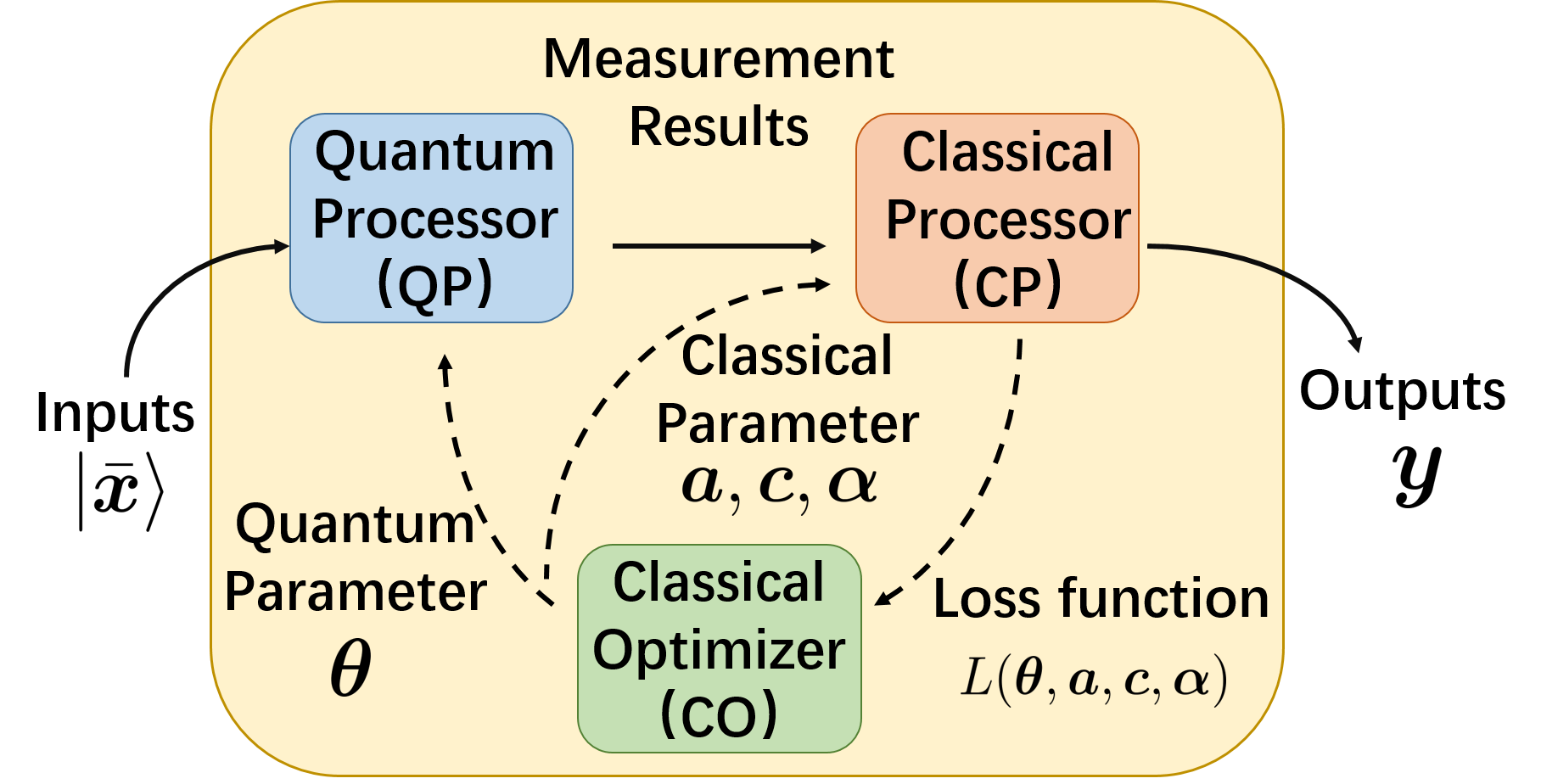}\label{frame}}
    \subfigure[]{\includegraphics[height=2.7cm,width=3.2cm]{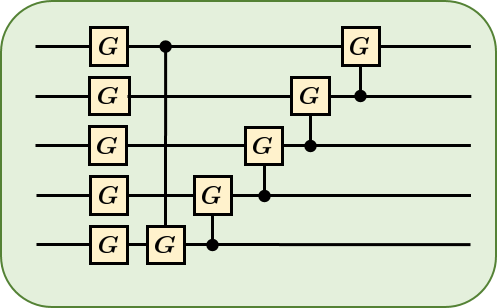}\label{ansatz}}
    \\
    \subfigure[]{\includegraphics[height=2.4cm,width=8cm]{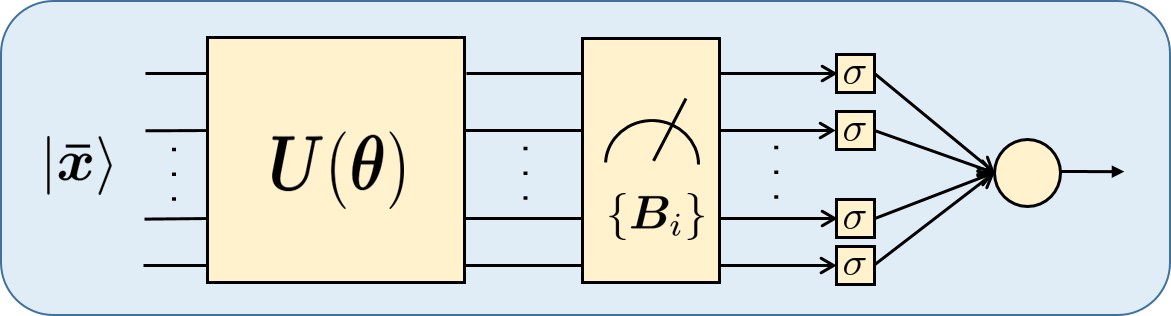}\label{circuit}}
    
    \caption{(a) The framework of the DQNN. The measurement results of the quantum processor are the inputs of the classical processor. The parameters are updated by a classical optimizer. (b) The circuit ansatz used in the numerical simulation. (c)The circuit structure of DQNN. The directed line represents the classical information, and the undirected line represents the quantum information.}
    \label{Fig1}
\end{figure}

 The specific structure of QP and CP are shown in Fig.~\ref{circuit}. In this paper, the QP part uses a specific circuit ansatz which is presented in~\cite{schuld2020circuit}. As shown in Fig.~\ref{ansatz}, the circuit ansatz represents the $U(\theta)$ in Fig.~\ref{circuit}. It contains $n$ parameterized single-qubit gates, $G(\theta_1,\theta_2,\theta_3)$, and $n$ parameterized two-qubit control gates, $CG(\theta_1,\theta_2,\theta_3):=|0\rangle\langle0|\otimes\mathbb{I}+|1\rangle\langle1|\otimes G(\theta_1,\theta_2,\theta_3)$ where $G$ is written as
\begin{equation}
    G(\theta_1,\theta_2,\theta_3)=\left(
    \begin{matrix}
    e^{i\theta_2}\cos(\theta_1) & e^{i\theta_3}\sin(\theta_1)\\
    -e^{-i\theta_3}\sin(\theta_1) & e^{-i\theta_2}\cos(\theta_1)
    \end{matrix}
    \right).
\end{equation}
QP outputs the measurement results, $\langle B_i\rangle=\mathrm{Tr}(U(\bm{\theta})|\bar{\bm{x}}\rangle\langle \bar{\bm{x}}|{U(\bm{\theta})}^{\dagger}B_i)$ with $N$ observables. The number, $N$, depends on the problem itself, and $\{B_i\}_{i=1}^{N}$ is a subset of the Pauli basis $\{P_i\}_{i=1}^{4^n}$.

The CP part applies the parameterized sigmoid function $\sigma^{(i)}$ to each of $\langle B_i\rangle$ where the sigmoid function is defined as
\begin{equation}
    \sigma^{(i)}(\langle B_i\rangle)\equiv\frac{1}{1+\exp\{-(a^{(i)}(\langle B_i\rangle-c^{(i)}))\}}
\end{equation}
with $a^{(i)}>2$ and $c^{(i)}\in[0,1]$. Then CP feeds $\{\sigma^{(i)}(\langle B_i\rangle)\}_{i=1}^N$ into a classical linear node and results in
\begin{equation}
    Q(\bar{\boldsymbol{x}}):=\sum_{j=1}^{N}\alpha_j\sigma(a_j(\langle B_j(\bar{\bm{x}},\bm{\theta})\rangle-c_j)),
\end{equation}
where $\bm{\alpha}=\{ \alpha_j\}$ are trainable. 

The final part of the DQNN, CO, minimizes a loss function $L(\boldsymbol{\theta},\boldsymbol{\alpha},\boldsymbol{a}, \boldsymbol{c}):=\sum_{k=1}^{m}||Q(\bar{\boldsymbol{x}}_k)-y_k||$ by using some gradient-based methods such as SGD~\cite{10.1007/978-3-7908-2604-3_16}, ADAM~\cite{DBLP:journals/corr/KingmaB14} and BFGS~\cite{buckley1985algorithm}, and obtains the optimal parameters, $(\boldsymbol{\theta}^*,\boldsymbol{\alpha}^*,\boldsymbol{a}^*, \boldsymbol{c}^*)=\arg\min L(\boldsymbol{\theta},\boldsymbol{\alpha},\boldsymbol{a}, \boldsymbol{c})$. The gradient of each parameter is analytically given below:
\begin{align*}
    \frac{Q(\bar{\bm{x}})}{\partial \theta_j}&=\frac{1}{2}\sum_i\alpha_i a_i\sigma^{(i)}(1-\sigma^{(i)})(\langle B_i\rangle_j^+-\langle B_i\rangle_j^-),\\
    \frac{Q(\bar{\bm{x}})}{\partial a_j}&=\alpha_j\sigma^{(j)}(1-\sigma^{(j)})(\langle B_j\rangle-c_j),\\
    \frac{Q(\bar{\bm{x}})}{\partial c_j}&=-\alpha_j\sigma^{(j)}(1-\sigma^{(j)})a_j,\\
     \frac{Q(\bar{\bm{x}})}{\partial \alpha_j}&=\sigma^{(j)}.
\end{align*}
where $\sigma^{(i)}$ represents $\sigma(a_i(\langle B_i(\bar{\bm{x}},\bm{\theta})\rangle-c_i))$ and  $\langle B_i\rangle_j^+$ and$\langle B_i\rangle_j^-$ denotes the expectation value $\langle B_i\rangle$ inserting $\pm\frac{\pi}{2}$ into the $j$-th quantum parameter $\theta_j$ according to the parameter-shift rule~\cite{mitarai2018quantum}.

The basic idea to design such a structure is whenever the QNN is designed by using the variational quantum circuit, the classical part containing in it can generate and enhance the nonlinearity of the hybrid system and further decrease the number of required qubits to achieve the quantum neural network. In this way, the universality can be proved and the complexity can be reduced.

\section{Universality of DQNN} 
The universality of DQNN is guaranteed by the following theorem: 
\begin{theorem}\label{theorem1}
    Given $f(\bar{\boldsymbol{x}}) \in L^2(\mathbb{S}_0)$, for arbitrary small $\epsilon$, we can select appropriate $N\in \mathbb{N}$, the unitary $U(\bm{\theta})$, observables $B_i$, and parameters $\alpha_i\in\mathbb{R}$, $a_i\in\mathbb{R}_+$ and $c_i\in [0,1](i=1,...,N)$ such that
    \begin{equation}
        \int_{\mathbb{S}_0} \left|\sum_{i=1}^{N}\alpha_i\sigma(a_i\langle B_i(\bar{\boldsymbol{x}},\boldsymbol{\theta})\rangle-c_i)-f(\bar{\boldsymbol{x}})\right|^2 d\mu(\bar{\boldsymbol{x}}) < \epsilon.
    \end{equation}
\end{theorem}

\begin{proof}
Denoting the quantum circuit of the DQNN in Fig.~\ref{circuit} as $U(\bm{\theta})$, which maps the quantum data $\ket{\bar{\boldsymbol{x}}}$ into $\ket{\bar{\boldsymbol{x}}_f}=U(\bm{\theta})\ket{\bar{\boldsymbol{x}}}$, the output $y$ of DQNN is derived through measuring a set of observables $\{B_i\}_{i=1}^{N}$ on the final state $\ket{\bar{\boldsymbol{x}}_f}$. Based on measurement statistics, $y=\sum_{i=1}^{N}\alpha_i\sigma(a_i(\langle B_i(\bar{\bm{x}},\bm{\theta})\rangle-c_i))$ is found according to:
\begin{align*}
\langle B_i(\bar{\boldsymbol{x}},\boldsymbol{\theta})\rangle=\bra {\bar{\boldsymbol{x}}_f}B_i\ket{\bar{\boldsymbol{x}}_f}=\sum_j\lambda_{i,j}|\langle\bar{\boldsymbol{x}}|\boldsymbol{\xi}_{i,j}\rangle|^2
\end{align*}
where $\lambda_{i,j}$ and $|\bm{b}_{i,j}\rangle$ are the eigenvalues and the eigenvectors of $B_i$, and $\ket{\bm{\xi}_{i,j}}\equiv U^{\dag}(\bm{\theta})|\bm{b}_{i,j}\rangle$. In the following, we make two assumptions: the number of observables $N$ is sufficiently large, and the circuit ansatz, represented as $U(\bm{\theta})$, should have enough expressibility to approximate any arbitrary unitary evolution. The former is guaranteed if the local Pauli operators on each qubit can be measured, and the latter is guaranteed if the set of physically-implementable gates form a universal gate set. Without loss of generality, we will restrict ourselves to the case of $\lambda_{i,1}=1$ and $\lambda_{i,j}=0$ ($j\neq 1$) where $\langle B_i\rangle = |\langle\bar{\boldsymbol{x}}|\boldsymbol{\xi}_i\rangle|^2$. The output of the DQNN is the function 
\[
q(\bar{\boldsymbol{x}})_{\bm{\alpha,a,\xi_1,\cdots, \xi_{N}, c}} = \sum_{i=1}^{N}\alpha_i\sigma(a_i(|\langle\bar{\boldsymbol{x}}|\boldsymbol{\xi}_i\rangle|^2-c_i)).
\]
Denoted by $Q(\mathbb{S}_0)$, the function space is composed of the finite linear combination of the sigmoid-type functions: 
\begin{equation}
    \begin{aligned}  & Q(\mathbb{S}_0) = \Big\{ q(\bar{\boldsymbol{x}})_{\bm{\alpha,a,\xi_1,\cdots, \xi_{N}, c}} \ : \  N\in\mathbb{N}, \bm{\alpha} \in \mathbb{R}^{N} ,\\
    & \quad \bm{a} \in\mathbb{R}_+^{N}, \bm{c} \in [0,1]^{N}, \{\bm{\xi}_i \}_{i=1}^{N} \subset \mathbb{S}_{\mathbb{C}^N} \Big\}. 
    \end{aligned}
\end{equation}
To prove the universality, it suffices to verify that $Q(\mathbb{S}_0)$ is dense in $L^2(\mathbb{S}_0)$. 
We assume that the closure $\overline{Q(\mathbb{S}_0)}\neq L^2(\mathbb{S}_0)$. 
The contradiction is shown in supplemental material. 
\end{proof}

One can see that the two above assumptions are crucial for the validity of the proof. In order to satisfy them, the circuit can become very long and is required to be repeated many times to derive the measurement outcomes for all $B_i$.
In theory, the classical NN requires infinite neurons in one layer to be universal~\cite{citeulike:3561150}. However, in practice, the classical case only uses a finite number of neurons and has obtained excellent results in various aspects. Analogously, for a given data set, a chosen DQNN structure, a finite number of observables in DQNN are often sufficient to approximate the target function well, and this will be demonstrated in the following examples. In addition, the length of the DQNN circuit depends on the given data set, and it may become very long for a special data set. In fact, it is an important unsolved question on how exactly the complexity of the DQNN circuit depends on the data set. 

\section{Circuit complexity for DQNN}
The advantage of DQNN is introducing the classical sigmoid function to generate the nonlinearity and significantly reduces the circuit complexity compared with two duplication-based QNNs, quantum circuit learning(QCL) algorithm~\cite{mitarai2018quantum} and circuit-centric quantum classifiers(CCQ) algorithm~\cite{schuld2020circuit}. The circuit complexity to implement the QNNs is given as $ C:=O(n_gn_b)$ where $n_g,n_b$ respectively denote the number of quantum gates and the number of measurement observables. Without loss of generality, we assume the number of gates is polynomial to the number of required qubits.

To show the difference of complexity among QNN, QCL and CCQ, we concentrate on a polynomial function approximation problem whose goal is approximating an $M$-order polynomial function of $\bm{x}\in\mathbb{R}^d$. The CCQ stores $\bm{x}$ into the amplitudes of data qubits as $|\bm{x}\rangle=\frac{1}{||\bm{x}||}\sum_{i=1}^d x_i|i\rangle$ and needs $O(M)$ copies of $|\bm{x}\rangle$. After applying $n_g=O(poly(M\lceil \log d \rceil))$ gates, $O(1)$ POVM operators is used to measure the system. Its complexity is $C_{\text{CCQ}}=O(poly(M\lceil \log d \rceil))$. In the meanwhile, the QCL encodes $\bm{x}$ as $\rho^{\otimes M}(\boldsymbol{x})=\frac{1}{2^d}\otimes_{i=1}^d \left(\otimes_{k=1}^M[\mathcal{I}+x_i\sigma_x^{(k)}+\sqrt{1-x_i^2}\sigma_z^{(k)}]\right)$ using $M$ copies of the data qubits $\rho(\bm{x})$. QCL uses $O(1)$ Pauli operators to measure the circuit. Therefore, its complexity is $C_{\text{QCL}}=O(poly(Md))$. Because the classical sigmoid function makes DQNN need no duplication and the encoding method is the same as CCQ, the DQNN only needs $n_g=O(poly(\lceil \log d \rceil))$ gates. Moreover, the number of observables is independent to the number of required qubits but a hyperparameter determined by the problem. The complexity of DQNN is $C_{\text{DQNN}}=O(poly(\lceil \log d \rceil))$. The comparison is summarized in Table.~\ref{compare}. It can be seen that DQNN efficiently reduces the number of required qubits compared with CCQ and QCL. 
\begin{table}[htp]
    \centering
    \begin{tabular}{c|c|c}
    \hline
    \hline
    Algorithm &  \# Duplication  & \# Data qubits \\
    \hline
    QCL & $M$ & $d$\\
    \hline
    CCQ & $M$ & $\lceil \log d \rceil$ \\
    \hline
    DQNN & $1$ & $\lceil \log d \rceil$ \\
    \hline
    \hline
    \end{tabular}
    \caption{The number of required qubits among three proposals to approximate an $M$-order polynomial function of $\bm{x}\in\mathbb{R}^d$. }
    \label{compare}
\end{table}

\begin{table*}[htp]
    \setlength{\belowcaptionskip}{-3.8mm}
    \centering
    \begin{tabular}{c|c|c|c|c|c}
    \hline
    \hline
    \multicolumn{6}{c}{Regression Task}\\
    \hline
    Algorithm & Qubits Number & Layers Number & Copies Number & $c$ &Mean Relative Error \\
    \hline
    DQNN & 2 & 1 & 1 & 120 & 4.29\%\\
    \hline
    QCL & 2 & 2 & 1 & 24 & 65.27\%\\
    \hline
    QCL & 4 & 3 & 2 & 144 &50.58\%\\
    \hline
    QCL & 6 & 6 & 3 & 648 & 87.51\%\\
    \hline
    \multicolumn{6}{c}{Classification Task}\\
    \hline
    Algorithm & Qubits Number & Layers Number & Copies Number & $c$ &Accuracy \\
    \hline
    DQNN & 2 & 1 & 1 & 240 & 97.63\%\\
    \hline
    CCQ & 2 & 3 & 1 & 78 & 56\%\\
    \hline
    CCQ & 4 & 3 & 2 & 300 & 81.25\%\\
    \hline
    CCQ & 6 & 6 & 3 & 1674 & 82.63\%\\
    \hline
    \hline
    \end{tabular}
    \caption{The running result and specific setup of DQNN, QCL and QCL. $c$ is a specific value of the circuit complexity which is calculated by $c=n_g*n_b$.}\label{fitting_result}  
    \qquad
    \vspace{2mm}
    
    \begin{tabular}{c|c|c|c|c|c}
    \hline
    \hline
    Task & n & $N_{train}$ & $N_{test}$ & $\epsilon_{Train}$ & $\epsilon_{Test}$ \\
    \hline
    MNIST$_2$ & $8$ &$12665$ & $2115$  & $0.0047$ & $0.0009$\\
    \hline
    MNIST$_3$ & $8$ &$18623$& $3147$ & $0.0172$ & $0.0114$\\
    \hline
    Wine & $4$ & $143$& $35$ & $0.0000$ & $0.0286$ \\
    \hline
    Breast Cancer&$4$&$560$& $39$ &$0.0143$ & $0.0432$ \\
    \hline
    \hline
    \end{tabular}
    \caption{Implement the DQNN on the real-world data sets. We use ADAM algorithm to optimize the parameters. $n$ indicates the number of qubits, $N_{Train}$ and $N_{Test}$ respectively represent the number of the training set and the test set. $\epsilon_{Train}$ and $\epsilon_{Test}$ represent the errors on each data set.}
    \label{Actual set}
\end{table*}

\section{Applications}\label{application_sec}
We design two data sets, regression and classification data sets, to show the advantage of the DQNN compared with QCL and CCQ. The regression data set (Fig.~\ref{fit_com_dataset}) contains $400$ data samples which are randomly generated by $f=(0.715625 - 1.0125 x_1^2 + x_1^4) (0.715625 - 1.0125 x_2^2 + x_2^4)$ with $x_1,x_2\in[-0.8,0.8]$. The classification data set (Fig.~\ref{classification_com_dataset}) has $800$ data samples where the boundaries are generated by $x_1^2+x_2^2=0.16$ and $x_1^2+x_2^2=0.81$ with $x_1,x_2\in[-1,1]$. In the donut-like area, the data samples are labeled $y^{(i)}=[1,0]^T$ and others are $y^{(i)}=[0,1]^T$. We use different numbers of duplications and layers in QCL and CCQ. The DQNN uses $5$ and $10$ randomly generated observables respectively in the regression and classification task. The simulation results (Table.~\ref{fitting_result}) show that the performances of DQNN are better than other proposals which complexity is around $5$ times larger both in regression and classification. With the increasing of the number of copies and layers, CCQ in the classification task shows underfitting which leads to a lower accuracy.
\begin{figure}[htp]
    \centering
     \subfigure[] 
    {
     \begin{minipage}{.45\linewidth}\label{fit_com_dataset}
      \includegraphics[width=4cm,height=3cm]{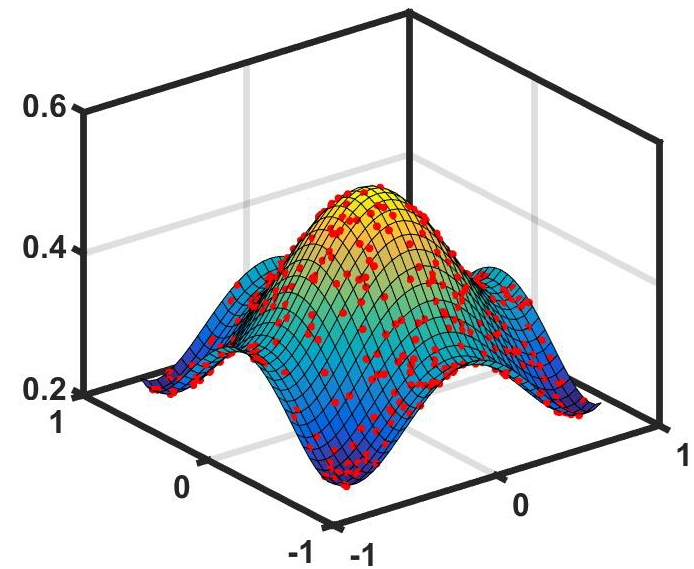}
     \end{minipage}
    }
    \subfigure[]
    {
     \begin{minipage}{.45\linewidth}\label{classification_com_dataset}
      \includegraphics[width=4cm,height=3cm]{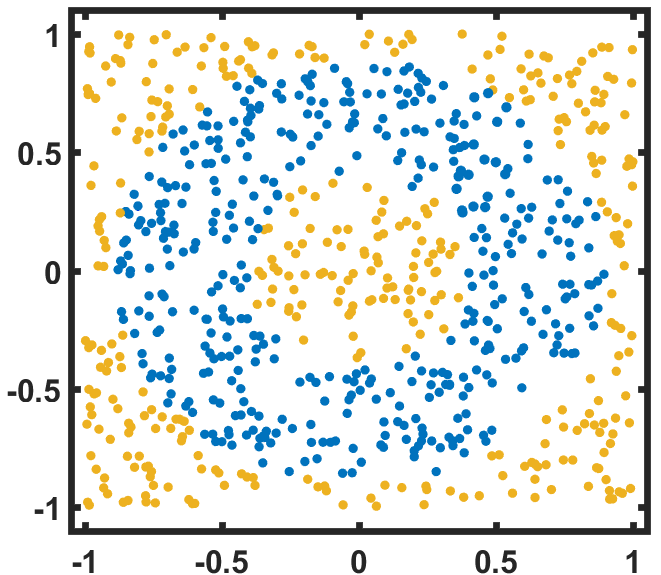}
     \end{minipage}
    }
   \caption{$(a)$ The regression data set with a polynomial function. $(b)$ The classification data set with two circular decision boundaries.}
\end{figure}

Additionally, in order to compare with the classical counterpart, we implement the classical neural network (NN) based on above two tasks to compare the performance with QNN. The results show that with similar number of parameters, the DQNN has the same power as its classical analogue. In the regression task, the training process (Fig.~\ref{fit_com_hybrid}) shows that both proposals have the similar rate of convergence and the mean relative error on NN is $4.19\%$ which is similar with DQNN. As for the classification task, the training process with the accuracy shown in. Fig.~\ref{classification_com_hybrid} demonstrates that the neural network with the similar number of parameters achieves $74.5\%$ (NNv2). Meanwhile, the classical neural network which contains three times more parameters than DQNN can achieves the similar accuracy as DQNN. 

To verify the power of DQNN on the real-world data sets, we further implement some classification tasks. Firstly, on the handwritten digits data set, MNIST~\cite{lecun2010}, we choose $0$ and $1$ for a binary classification and $0$, $1$, $2$ for a multi-target classification on MNIST. Each picture is reshaped into $16\times16$ and encoded into $8$ qubits. Additionally, on the Wine and Breast Cancer data sets~\cite{Dua:2019}, we randomly divide the data sets into five equivalent part and pick one of them as the test set. The result (Table.~\ref{Actual set}) shows the DQNN has the ability to find the complexity pattern hidden in the real-world data sets.   

Besides the classical tasks above, DQNN provides the ability to investigate the intrinsic property of quantum mechanics such as the quantum phase recognition (QPR). Specifically, we apply the DQNN to the $\mathbb{Z}_2 \times \mathbb{Z}_2$ symmetry-protected topological (SPT) phase discrimination task~\cite{cong2019quantum}. The ground states of a parameterized spin$-1/2$ chain Hamiltonian,
\begin{align*}
    H=&-J\sum_{i=1}^{N-2}\sigma_z^{(i)}\sigma_x^{(i+1)}\sigma_z^{(i+2)}-h_1\sum_{i=1}^N\sigma_x^{(i)}\\
    &-h_2\sum_{i=1}^{N-1}\sigma_x^{(i)}\sigma_x^{(i+1)}
\end{align*}
where $h_1$, $h_2$ and $J$ are parameters, corresponds to the different topological phases. The phase diagram of the Hamiltonian is given in Fig.~\ref{phase diagram}. The train data takes $400$ equally spaced points from $h_1\in[0,1.6]$ and $h_2\in[-1.6,1.6]$. And the test data contains $4096$ equally spaced points. The ground state corresponding to each point is labeled $[1,0]^T$, if it belongs to the SPT phase. Otherwise it is labeled $[0,1]^T$. We numerically implement the DQNN with $15$ qubits, $420$ parameters and $10$ observables. The accuracy on the test data achieves $99.10\%$. It shows the DQNN could find the relation between the ground states of the Hamiltonian and their corresponding phase.
\begin{figure}[htp]
    \centering
    \subfigure[]
    {
     \begin{minipage}{.45\linewidth}\label{fit_com_hybrid}

     \includegraphics[width=4cm,height=3cm]{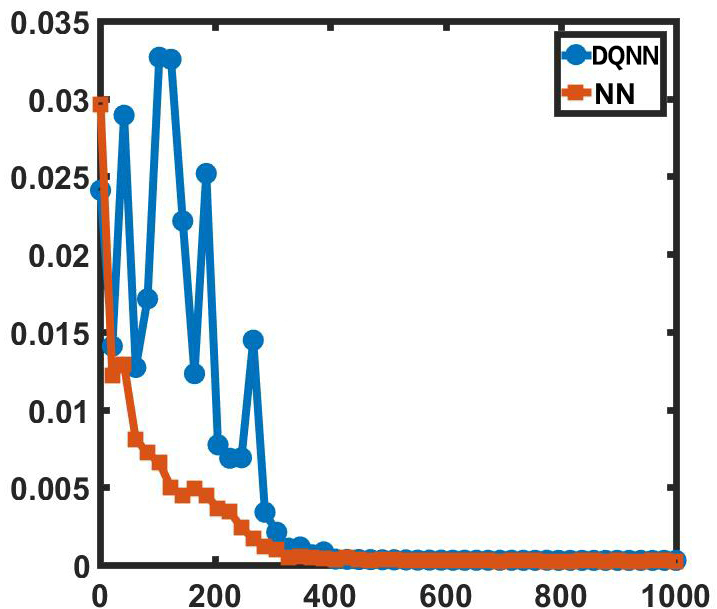}
     \end{minipage}
    }
    \subfigure[]
    {
     \begin{minipage}{.45\linewidth}\label{classification_com_hybrid}
     \includegraphics[width=4cm,height=3cm]{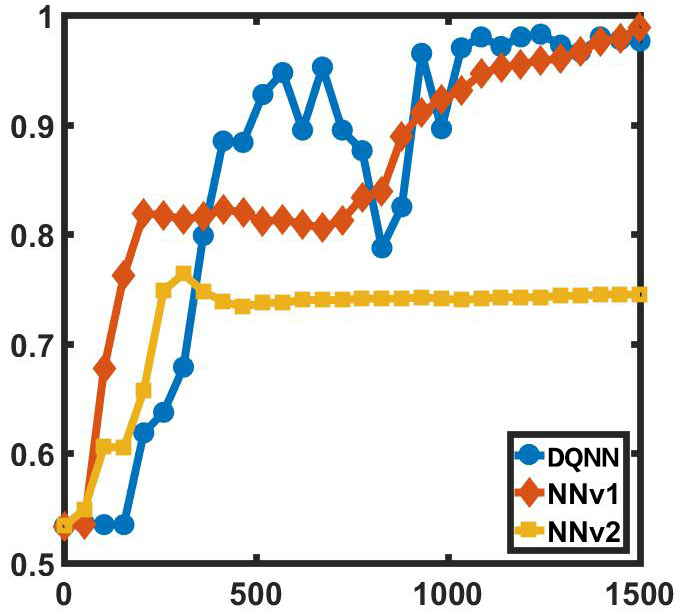}
     \end{minipage}
    }
    \subfigure[]
    {
     \begin{minipage}{\linewidth}\label{phase diagram}
     \centering
     \includegraphics[width=5.5cm,height=4cm]{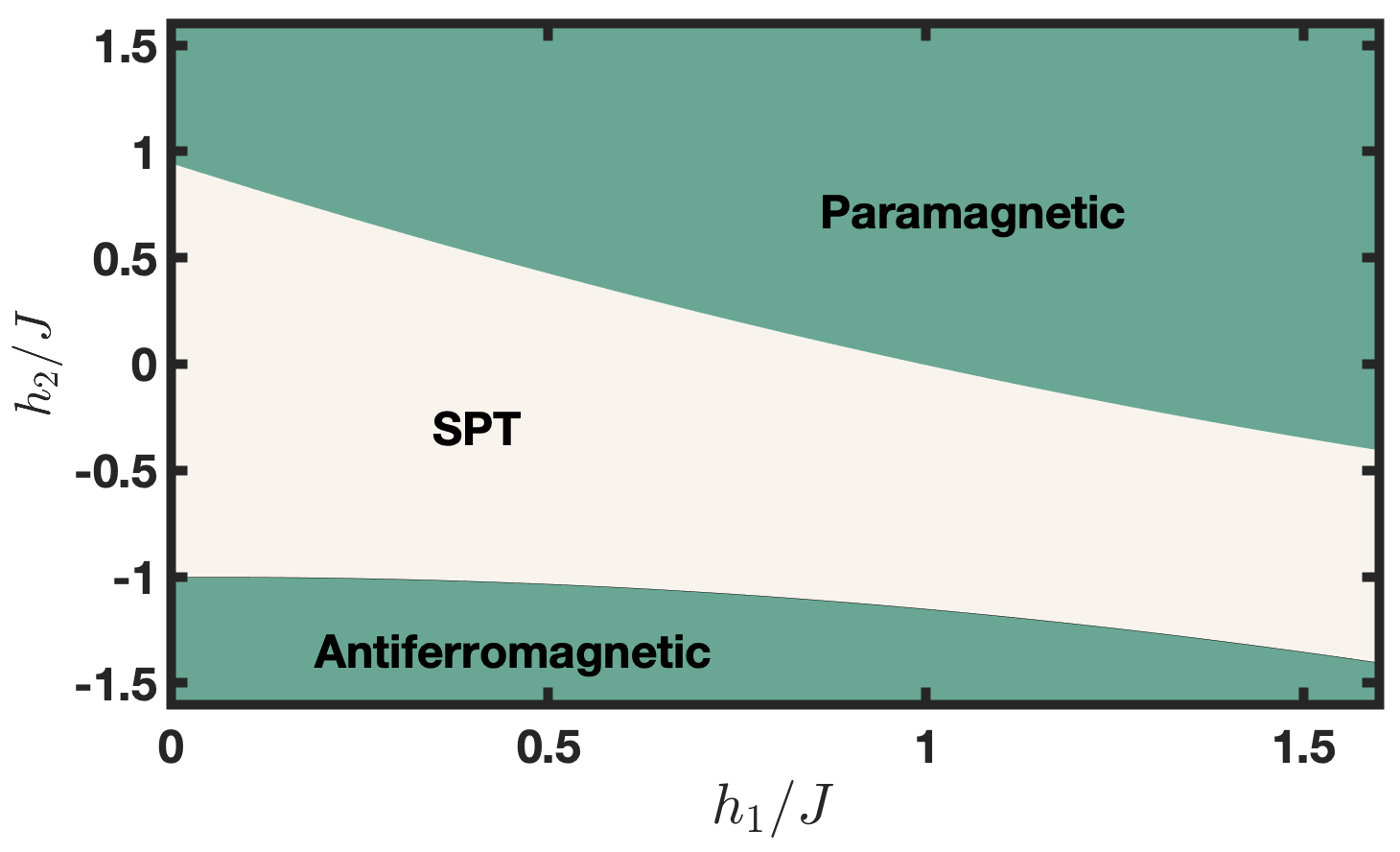}
     \end{minipage}
    }
   \caption{$(a)$ The MSE loss with training episodes of QNN and classical neural network. $(b)$ The classification accuracy with training episodes of QNN and classical neural network. $(c)$ The phase diagram of the spin$-1/2$ chain. The phase boundary is generated by the 2-degree polynomial regression based on some boundary points.}
   \label{Fig6}
\end{figure}

\section{Conclusion}
In this article, we present the universal duplication-free quantum neural network whose nonlinearity is generated by the classical sigmoid function. The simulation results show that the DQNN significantly reduces the number of required qubits to complete the supervised learning tasks compared with previous work and has the ability to recognize the SPT phase of a spin-$1/2$ chain Hamiltonian.  However, how to design an appropriate circuit ansatz for a certain problem still remains an open question. Besides the scenarios discussed in this work, we expect the duplication-free quantum neural network has broad applications in other area, including natural language processing, computer version and reinforcement learning.

\section{Acknowledgement}
The authors gratefully acknowledge the grant from National Key R\&D Program of China, Grant No. 2018YFA0306703. We also thank Chu Guo, Bujiao Wu, Yusen Wu, Shaojun Wu, Yuhan Huang, Dingding Wen and Yi Tian for helpful discussions.
\bibliographystyle{unsrtnat}
\bibliography{ref.bib}

\onecolumn\newpage
\appendix

\section{The proof of Theorem 1}\label{appendix1}
It suffices to show that $Q(\mathbb{S}_0)$ is dense in $L^2(\mathbb{S}_0)$. We assume that the closure $\overline{Q(\mathbb{S}_0)}\neq L^2(\mathbb{S}_0)$ and show the contradiction. By the Hahn-Banach theorem~\cite{hahn1927lineare}, there exists a bounded linear functional $\mathcal{L}$ of $L^2(\mathbb{S}_0)$ such that $\mathcal{L}(Q(\mathbb{S}_0))=0$ and $\mathcal{L}\neq 0$. 
By the Riesz representation theorem~\cite{frechet1907ensembles}, there exists a function $g(\bar{\boldsymbol{x}})\in L^2(\mathbb{S}_0)$ such that
    \begin{equation}
        \mathcal{L}(f)=\int_{\mathbb{S}_0}f(\bar{\boldsymbol{x}})g(\bar{\boldsymbol{x}})d\mu(\bar{\boldsymbol{x}})\quad \text{for all} f\in L^2(\mathbb{S}_0)
    \end{equation}
    where $\mathcal{L} \neq 0$ implies that $g(\bar{\boldsymbol{x}}) \neq 0$. Since $\mathcal{L}(Q(\mathbb{S}_0))=0$, we have
    \begin{equation}
        \int_{\mathbb{S}_0} \sigma(a(|\langle\bar{\boldsymbol{x}}|\boldsymbol{\xi}\rangle|^2-c))g(\bar{\boldsymbol{x}})d\mu(\bar{\boldsymbol{x}})=0 
    \end{equation}
    In particular, there exists an open subset $E \subset \mathbb{S}_0$ with the measure $\mu(E)>0$ such that $g(\bar{\boldsymbol{x}}) \neq 0$ in $E$. Without loss of generality, we assume $g(\bar{\boldsymbol{x}}) \geq k>0$ in $E$. Since $E$ is open, there exists a small ball $B(\boldsymbol{\xi}^*,\delta)=\{\bar{\boldsymbol{x}}:|\bar{\boldsymbol{x}}-\boldsymbol{\xi}^*|<\delta\}\subset E$. 
    
For $\boldsymbol{\xi}^*$, $\bar{\boldsymbol{x}}\in \mathbb{S}_{\mathbb{R}^N}$, we have 
    \begin{equation}
        |\bar{\boldsymbol{x}}-\boldsymbol{\xi}^*|^2 =2-2\langle\bar{\boldsymbol{x}}|\boldsymbol{\xi}^*\rangle.  
    \end{equation}
Therefore $|\langle\bar{\boldsymbol{x}}|\boldsymbol{\xi}^*\rangle|^2>c$ is equivalent to 
    \begin{equation}\label{Eq:x-xi-x+xi}
        |\bar{\boldsymbol{x}}-\boldsymbol{\xi}^*|^2<2(1-\sqrt{c})\quad\text{or}\quad|\bar{\boldsymbol{x}}-\boldsymbol{\xi}^*|^2>2(1+\sqrt{c}).
    \end{equation}
We claim that $|\bar{\boldsymbol{x}}-\boldsymbol{\xi}^*|^2>2(1+\sqrt{c})$ is impossible to hold if $c$ is closed to $1$. 
Or else, by using
\begin{equation}
    |\bar{\boldsymbol{x}}-\boldsymbol{\xi}^*|^2+|\bar{\boldsymbol{x}}+\boldsymbol{\xi}^*|^2=4
\end{equation} 
we find that $|\bar{\boldsymbol{x}}+\boldsymbol{\xi}^*|^2<2(1-\sqrt{c})$. From Eqn.~(\ref{Eq:x-d+1}) and $\boldsymbol{\xi}^*$, $\bar{\boldsymbol{x}}\in \mathbb{S}_0$, we see that $|\bar{\boldsymbol{x}}+\boldsymbol{\xi}^*|^2\geq (\xi^*_{d+1}$ + $\bar{x}_{d+1})^2\geq 4(1+(1+\kappa_2)^2)^{-1}$. 
Hence for $c>(1-2(1+(1+\kappa_2)^2)^{-1})^2$, the latter case of Eqn.~(\ref{Eq:x-xi-x+xi}) makes no sense, and $|\langle\bar{\boldsymbol{x}}|\boldsymbol{\xi}^*\rangle|^2>c$ is only equivalent to 
\begin{equation}
    |\bar{\boldsymbol{x}}-\boldsymbol{\xi}^*|^2<2(1-\sqrt{c})\quad\forall \boldsymbol{\xi}^*,\bar{\boldsymbol{x}}\in\mathbb{S}_0.
\end{equation}
Therefore, passing to the limit $a\rightarrow\infty$ we obtain 
\begin{equation}
    \sigma(a(|\langle\bar{\boldsymbol{x}}|\boldsymbol{\xi}^*\rangle|^2-c))\rightarrow \left\{
    \begin{aligned}
        &1 \quad \forall \bar{\boldsymbol{x}}\in B(\boldsymbol{\xi}^*,\delta_1)\\
        &0 \quad \forall \bar{\boldsymbol{x}}\notin B(\boldsymbol{\xi}^*,\delta_1)
    \end{aligned}
\right.
\end{equation}
with $\delta_1=(2(1-\sqrt{c}))^{\frac{1}{2}}$. By taking $c$ sufficiently close to 1 such that $\delta_1\leq\delta$, and using Lebesgue dominate convergence theorem, from Eqn.(B13) we obtain
    
\begin{align}
    0&=\int_{\mathbb{S}_0} \sigma(a(|\langle\bar{\boldsymbol{x}}|\boldsymbol{\xi}\rangle|^2-c))g(\bar{\boldsymbol{x}})d\mu(\bar{\boldsymbol{x}})\\
    &\geq k\int_{B(\boldsymbol{\xi}^*,\delta_1)}\sigma(a(|\langle\bar{\boldsymbol{x}}|\boldsymbol{\xi}\rangle|^2-c))d\mu(\bar{\boldsymbol{x}})\\
    &\rightarrow k\int_{B(\boldsymbol{\xi}^*,\delta_1)}1d\mu(\bar{\boldsymbol{x}})=k\mu(B(\boldsymbol{\xi}^*,\delta_1))>0
\end{align}
 which comes out a contradiction. Hence, we conclude that $Q(\mathbb{S}_0)$ is dense in $L^2(\mathbb{S}_0)$. Thus for any $f\in L^2(\mathbb{S}_0)$ and $\epsilon>0$, we can find a $q(\bar{\boldsymbol{x}}) \in Q(\mathbb{S}_0)$ such that
\begin{equation}
    ||f-q(\bar{\boldsymbol{x}})||^2_{L^2(\mathbb{S}_0)}=\int_{\mathbb{S}_0}|f(\bar{\boldsymbol{x}})-q(\bar{\boldsymbol{x}})|^2d\mu(\bar{\boldsymbol{x}})\leq \epsilon
\end{equation}
which proves the theorem.

\end{document}